\documentclass[12pt]{iopart}

\usepackage{iopams}  

\usepackage{graphicx,color}
\usepackage{bm} 

\usepackage{amsthm}

\definecolor{dgreen}{rgb}{0,0.5,0}

\def\cA{{\mathcal{A}}}
\def\cB{{\mathcal{B}}}
\def\cC{{\mathcal{C}}}

\def\cS{{\mathcal{S}}}
\newcommand{\ket}[1]{|{#1}\rangle}
\newcommand{\bra}[1]{\langle{#1}|}

\renewcommand{\Re}{\mathop{\mathrm{Re}}\limits}

\newtheorem{thm}{Theorem}
\newtheorem{lem}{Lemma}

\newtheorem{prop}{Proposition}

\theoremstyle{definition} 
\newtheorem{defn}{Definition}

\theoremstyle{remark} 
\newtheorem*{rem}{Remark}

\begin{document}
\title[Entanglement and Quantumness]{Quantumness and Entanglement Witnesses}
\author{Paolo Facchi$^{1,2}$, Saverio Pascazio$^{3,2}$, Vlatko Vedral$^{4,5,6}$ and  Kazuya Yuasa$^{7}$}
\address{$^{1}$Dipartimento di Matematica and MECENAS, Universit\`a di Bari, I-70125  Bari, Italy}
\address{$^{2}$INFN, Sezione di Bari, I-70126 Bari, Italy}
\address{$^{3}$Dipartimento di Fisica and MECENAS, Universit\`a di Bari, I-70126  Bari, Italy}
\address{$^{4}$Department of Atomic and Laser Physics, Clarendon Laboratory, University of Oxford, Oxford OX1 3PU, UK} 
\address{$^{5}$Centre for Quantum Technologies, National University of Singapore, Singapore 117543, Singapore}
\address{$^6$Department of Physics, National University of Singapore, Singapore 117542, Singapore}
\address{$^{7}$Waseda Institute for Advanced Study, Waseda University, Tokyo 169-8050, Japan}
\date{\today}

\begin{abstract}
We analyze the recently introduced notion of quantumness witness and compare it to that of entanglement witness. We show that any entanglement witness is also a quantumness witness. We then consider some physically relevant examples and explicitly construct some witnesses. 
\end{abstract}


\maketitle

\section{Introduction}
The quantal features of a physical system pertain both to its states and observables. States evolve according to the Schr\"odinger equation and admit a probabilistic interpretation. Observables make up an algebra of operators, that in general do not commute and cannot be simultaneously measured. Defining quantumness and classicality is an interesting and subtle problem, that can be tackled from different perspectives, both in physics and mathematics. 

Composed quantum systems, made up of two or more subsystems, can be entangled. Entanglement is a very peculiar quantum characteristic and has become an important resource in quantum information and quantum applications. 
Both concepts of entanglement and quantumness are often investigated by framing them in terms of inequalities: entanglement and separability are discriminated through the Bell inequality \cite{bell}, while quantumness and classicality are discriminated through the Leggett-Garg inequality {\cite{LG,Leggett-JPC2002}. 

A recent attempt in the study of quantumness and classicality has been made by Alicki and collaborators \cite{AVR,APVR}, who introduced the idea of ``quantumness witness,'' motivating interesting experiments \cite{brida08,bridagisin}. Both the theoretical proposal and the experiments mainly focused on single qubits, in the attempt to test their quantum features and rule out (semi)classical descriptions.

In this article we shall adopt this approach by focusing on composed systems. We shall propose a combined framework by casting these notions in terms of mathematical definitions and the idea of witnesses. In particular, we shall show that any entanglement witness is also a quantumness witness. 
To this end, we shall make use of a fully algebraic approach \cite{Araki,BratteliRobinson}.

We shall start by introducing notation and definitions in Sec.\ \ref{clq}. The definitions we shall propose are somewhat more general than those of Refs.\  \cite{AVR,APVR}. 
We show in Sec.\ \ref{alleq} that any entanglement witness is a quantumness witness. 
Some explicit examples will be worked out in Sec.\ \ref{explicitc}, where we look in particular at the Bell inequality. We conclude with a few remarks in Sec.\ \ref{concl}.

\section{Classicality, quantumness and entanglement}
\label{clq}
We introduce notation and define quantumness and entanglement witnesses.
We shall only consider finite dimensional systems.

\subsection{Quantumness witnesses}
\label{QW}
We have the following characterization of commutative (i.e.\ classical) algebras.
\begin{thm}[\cite{AVR,APVR}]
\label{thm:abelian}
Given a $C^*$-algebra $\mathcal{A}$, the following two statements are equivalent:
\begin{enumerate}
\item
\label{classicalC} 
$\mathcal{A}$ is commutative.
To wit, for any pair $X, Y\in\mathcal{A}$,
\begin{equation*}
[X,Y] :=XY-YX=0.
\end{equation*}
\item
\label{classicalB}
For any pair $X, Y\in\mathcal{A}$ with  $X\geq 0$ and $Y\geq 0$,
\begin{equation*}
\{X,Y\}:=XY + YX \geq 0.
\end{equation*}
\end{enumerate}
\end{thm}

As a consequence, for a quantum system one can always find pairs of observables 
$X\geq 0$, $Y\geq 0$ such that the observable 
\begin{eqnarray}
\label{wqsym}
Q_{\mathrm{AVR}} = \{X,Y\} 
\end{eqnarray}
is \emph{not} positive semidefinite.
Thus, $Q_{\mathrm{AVR}}\in\cA$ is a ``witness" of the quantumness (i.e.\ noncommutativity) of the algebra $\cA$ \cite{AVR,APVR}.

We  define classical states, a concept that will be useful in the following.
\begin{defn}
We say that a state $\rho\in\cS(\cA)$ is \emph{classical} if 
\begin{equation}
\rho([X,Y])=0, 
\quad\mbox{for any pair}\   X, Y\in\cA .
\label{rhoclassical}
\end{equation}
A state that is not classical is \emph{quantum}.
\end{defn}
\begin{rem}
We recall that the set $\cS$ of states  of a given algebra $\cA$ is the subset of the continuous linear complex functionals $\rho \in \cA^*$ (the dual space of $\cA$) that are positive and normalized, i.e.\ $\rho(A^* A)\geq 0$ for any $A\in\cA$ ($A^*$ being the adjoint), and $\rho(\mathbb{I})=1$.
See \cite{Araki,BratteliRobinson}.
\end{rem}
\begin{rem}
Let us also recall that (normal) states $\rho\in\cS$ can be \emph{uniquely} realized as traces over density matrices $\tilde{\rho}$ \emph{belonging} to the algebra $\cA$:
\begin{equation}
\rho(A) = \tr (\tilde{\rho} A), \qquad  \tilde{\rho}\in \cA, \quad  \tilde{\rho} \geq 0,\quad \tr\tilde{\rho}=1.
\label{rhodens}
\end{equation}
We warn the reader that in the following we will freely use  this identification and
commit the sin of not distinguishing between states and density matrices.
\end{rem}

Notice that we can have classical states even when the algebra is noncommutative (namely, even when there exist $A$ and $B$ such that $[A,B] \neq 0$). 
In words, classical states do not ``perceive'' nonvanishing commutators.
Moreover, the definition (\ref{rhoclassical}) of classical state is weaker than the notion of classicality that emerges from (\ref{classicalC})--(\ref{classicalB}) of Theorem~\ref{thm:abelian}. Indeed, $\cA$ is commutative iff  \emph{every} state $\rho\in\cS$ is classical.

\begin{rem}
Let us notice that, in general, mixtures are not classical states. For example, a qubit state $\rho=p\ket{0}\bra{0}+q\ket{1}\bra{1}$ is not classical, since it possesses coherence, e.g.\ $\bra{-}\rho\ket{+}=c_0c_1(p-q)$ for $\ket{+}=c_0\ket{0}+c_1\ket{1}$ and $\ket{-}=c_1^*\ket{0}-c_0^*\ket{1}$,  which is nonvanishing provided $p\neq q$ and $c_0,c_1\neq0$. On the other hand, the completely mixed state $\rho=\mathbb{I}/2$ is classical, in that it does not possess any coherence, $\bra{-}\rho\ket{+}=0$ for any $c_0$ and $c_1$.
\end{rem}

Let us now define quantumness witnesses.
\begin{defn}
We  say that an observable $Q\in\cA$ is a \emph{quantumness witness} (QW) if
\begin{enumerate}
\item for any classical state $\rho\in\cS$ one gets
$\rho(Q)\geq 0$,
\item there exists a (quantum) state $\sigma\in\cS$ such that 
$\sigma(Q) < 0$.
\end{enumerate}
\end{defn}
The fact that the particular observables $Q_{\mathrm{AVR}}$ in~(\ref{wqsym}) are QWs follows from the following lemma. 
\begin{lem}
For any classical state $\rho\in\cS$ and for any pair $X,Y\in\cA$ with $X\geq0$, $Y\geq 0$ it happens that
\begin{equation}
\rho (\{X,Y\}) \geq 0.
\end{equation} 
\end{lem}
\begin{rem}
In words, classical states do not even perceive the possible negativity of the anticommutators $\{X,Y\}$: their behavior is fair with respect to~(\ref{classicalC})--(\ref{classicalB}) of Theorem~\ref{thm:abelian}. 
\end{rem}
\begin{proof}
Since $\rho$ is classical we get
\begin{equation}
\rho(\{X,Y\})= \rho(2XY - [X,Y]) = 2 \rho(XY).
\end{equation}
Recall that an observable $X$ is nonnegative iff $X= A^* A$ for some $A\in\cA$.  Therefore,
\begin{equation}
\rho(XY)= \rho(A^* A B^* B)
\end{equation}
for some $A,B\in\cA$. By using again the definition of classicality (\ref{rhoclassical}) we conclude
\begin{equation}
\rho(XY) = \rho(B A^* A B^*) = \rho (C^* C) \geq 0,
\end{equation}
with $C= A B^* \in \cA$.
\end{proof}

\subsection{Entanglement witnesses}
Let our system be made up of two \emph{subsystems}, that will conventionally be sent to Alice and Bob, whose observations are independent. The notion of independence is reflected in the fact that the total algebra of observables is assumed to factorize in two subalgebras 
\begin{equation}
\mathcal{C}=\mathcal{A}\otimes\mathcal{B}.
\label{Atot}
\end{equation}
Namely, the two subalgebras commute with each other, but each subalgebra can be noncommutative.

\begin{defn}
A state $\rho\in\cS(\mathcal{C})$ is said to be \emph{separable} (with respect to the given bipartition $\cA\otimes\cB$)
if it can be written as a convex combination of product states, namely,
\begin{equation}
\rho= \sum_k p_k \rho_k \otimes \sigma_k, \qquad p_k>0, \qquad \sum_k p_k=1,
\label{sepdef}
\end{equation}
where $\rho_k\in\cS(\cA)$ and $\sigma_k\in\cS(\cB)$ are states of $\cA$ and $\cB$, respectively. A state that is not separable is said to be \emph{entangled} (with respect to the given bipartition).
\end{defn}

\begin{rem}
The definition of separability depends on the algebra $\cC$ of the composed system, that in general can be reducible, i.e.\ the matrices $C\in\mathcal{C}$ are block-diagonal, $C= \bigoplus_k C_k$. If states are identified with density matrices belonging to the algebra as in~(\ref{rhodens}), then they inherit the block-diagonal form of the latter.
\end{rem}

\begin{defn}[\cite{entanglementrev,entrevhoro}]
We  say that an observable $E\in\mathcal{C}$ is an \emph{entanglement witness} (EW) if
\begin{enumerate}
\item for any separable state $\rho\in\cS(\mathcal{C})$ one gets
$\rho(E)\geq 0$,
\item there exists a (entangled) state $\sigma\in\cS(\mathcal{C})$ such that 
$\sigma(E) < 0$.
\end{enumerate}
\end{defn}

\section{All EWs are QWs}
\label{alleq}
We now show that every EW is also a QW\@.
We first consider a preliminary lemma.
\begin{lem}
Any classical state is separable.
\end{lem}
\begin{proof}
Notice first that if the algebra $\cC=\cA\otimes \cB$  is the full algebra of operators 
\begin{equation}
\cC=B(\mathbb{C}^n)\otimes B(\mathbb{C}^m) ,
\label{eq:totA}
\end{equation}
then the only classical state is the totally mixed state,
\begin{equation}
\rho= \mathbb{I}_{nm}/nm= \mathbb{I}_n/n \otimes \mathbb{I}_m/m,
\label{eq:totmix}
\end{equation}
which is obviously separable.
In general, however, the (sub)algebras $\cA$ and $\cB$ are reducible (i.e.\ they are proper subalgebras of the full matrix algebra) and one has
\begin{equation}
\fl \qquad \cC = \Big(\bigoplus_k B(\mathbb{C}^{n_k}) \Big) \otimes \Big(\bigoplus_l B(\mathbb{C}^{m_l})\Big) =
\bigoplus_{k,l} B(\mathbb{C}^{n_k}) \otimes B(\mathbb{C}^{m_l})=: \bigoplus_{k,l} \cC_{kl},
\label{eq:dirsum}
\end{equation}
where each $\cC_{kl}$ is an irreducible algebra of dimension $n_k m_l$. All observables are block-diagonal and the classical states have the form
\begin{equation}
\rho = \bigoplus_{k,l} p_{kl} \mathbb{I}_{n_k}/{n_k} \otimes \mathbb{I}_{m_l}/{m_l},
\end{equation}
with $p_{kl}\geq 0$ and $\sum_{kl}p_{kl}=1$, i.e.\ they are separable.
\end{proof}
\begin{rem}
Notice that if the two subalgebras are reducible, states inherit their block-diagonal structure. See Remark after Eq.\ (\ref{sepdef}).
\end{rem}

Our main theorem is now an easy consequence of the lemma just proved.
\begin{prop}
Any EW is a QW\@.
\end{prop}
\begin{proof}
Consider an EW $E\in\cC$. By definition $\rho(E)\geq 0$ for any separable $\rho\in\cS$. But by the previous lemma all classical states are separable. It follows that $\rho(E)\geq 0$ for any classical state $\rho$. Moreover, by definition, $\sigma(E) <0$ for some entangled state $\sigma$, which by the previous lemma must be a quantum state. Thus, $E$ is a QW\@.
\end{proof}
\begin{rem}
The converse is, of course, not true. If the algebra $\cA$ is noncommutative, and $Q\in\cA$ is a QW of the quantum state $\sigma\in \cS(\cA)$, then
\begin{equation}
\tilde{Q}=Q\otimes \mathbb{I}\in \cC
\end{equation} 
is also a QW (of the total algebra), but it is \emph{not} an EW\@. Indeed,  it is negative on separable states of the form $\sigma\otimes \omega$ [for any $\omega \in \cS(\cB)$], namely,
\begin{equation}
(\sigma\otimes \omega) (\tilde{Q}) <0.
\end{equation}
\end{rem}

\section{Explicit construction of EWs as anticommutators for qudits}
\label{explicitc}
In the previous section we have shown that an EW is always a QW\@. In view of that, among all QWs, those of the simple form~(\ref{wqsym}), that we shall call \emph{anticommutator QWs}, are quite interesting for possible applications, for example for efficiently generating EWs out of anticommutators. Therefore, here we shall  investigate whether an EW  $E$ can be written in the \emph{particular} form (\ref{wqsym}), namely, whether there exists a pair of positive operators $X$ and $Y$ such that  $ E= \{X,Y\}$.
Before looking at a rather general case, we consider some instructive examples.

\subsection{Swap operator and Bell inequality}
An example of EW for a $d\times d$ system is the  swap operator \cite{entrevhoro}
\begin{equation}
S=\sum_{i,j=0}^{d-1}\ket{i}\bra{j}\otimes\ket{j}\bra{i}.
\label{eq:swapop}
\end{equation}
Here the algebra is the full algebra of matrices $B(\mathbb{C}^d)\otimes B(\mathbb{C}^d)$, and $\{\ket{j}\}_j$ is a chosen orthonormal basis of $\mathbb{C}^d$ (computational basis). $S$ is nonnegative, $\rho(S)\ge0$, for all separable states $\rho$, but  it possesses an eigenvalue equal to $-1$.

Another interesting example of EW is the Bell-CHSH observable
\begin{equation}
\label{Bell }
E_{\mathrm{Bell}} = 2 \pm (A_1\otimes B_1 + A_1\otimes B_2+A_2\otimes B_1 - A_2\otimes B_2),
\end{equation}
where $A_{1,2}\in\mathcal{A}$ and $B_{1,2}\in\mathcal{B}$ are dichotomic observables (with eigenvalues $\pm1$) of Alice and Bob, respectively, and $A_{1,2}^2=\mathbb{I}$, $B_{1,2}^2=\mathbb{I}$.
If $\rho(E_\mathrm{Bell})<0$, $E_\mathrm{Bell}$ witnesses the violation of the Bell-CHSH inequality in the entangled state $\rho$.

For instance, if we take 
\begin{eqnarray}
\label{bells}
A_1=\sigma_x,\qquad
B_1=\frac{1}{\sqrt{2}}(\sigma_x+\sigma_y),\nonumber\\
A_2=\sigma_y,\qquad
B_2=\frac{1}{\sqrt{2}}(\sigma_x-\sigma_y),
\end{eqnarray}
where $\sigma_{x,y,z}$ are Pauli operators
\begin{equation}
\sigma_x=\ket{0}\bra{1}+\ket{1}\bra{0},\ \ %
\sigma_y=-i\ket{0}\bra{1}+i\ket{1}\bra{0},\ \ %
\sigma_z=\ket{0}\bra{0}-\ket{1}\bra{1},
\end{equation}
then
\begin{equation}
E_\mathrm{Bell}=2\pm\sqrt{2}(\sigma_x \otimes \sigma_x +\sigma_y \otimes \sigma_y).
\label{eq:EBell1}
\end{equation}
Observe now that the swap operator (\ref{eq:swapop}) and the Bell-CHSH observable (\ref{eq:EBell1}) are related by
\begin{equation}
S=P_{00}+P_{11}\pm\frac{1}{2\sqrt{2}}(E_\mathrm{Bell}-2),
\label{S-Bell}
\end{equation}
where
\begin{equation}
P_{ij}=\ket{i}\bra{i}\otimes\ket{j}\bra{j}\qquad (i,j = 0,1)
\end{equation}
are projections.

Due to the negative shift $-2$ in (\ref{S-Bell}), $S$ is more efficient at witnessing entanglement than $E_\mathrm{Bell}$: $S$ can actually detect entangled states that do not violate the Bell inequality.
For instance, let $\ket{\pm}=(\ket{0}\pm\ket{1})/\sqrt{2}$, then the entanglement of the vector state
\begin{equation}
 \ket{\chi}=  a\ket{+}\otimes\ket{-}+b\ket{-}\otimes\ket{+}
\end{equation}
is witnessed by $S$ if $\Re(a^*b)<0$, 
while $E_{\mathrm{Bell}}$ in Eq.\ (\ref{eq:EBell1}) (with the $+$ sign) is negative only for $\Re(a^*b)<-(\sqrt{2}-1)/2$.

\subsection{Reviewing previous results}
We now briefly review some results obtained in Ref.\ \cite{APVR}. Let 
\begin{eqnarray}
X & = & 2 \pm (A_1\otimes B_1 + A_1\otimes B_2) \geq 0,\nonumber\\
Y & = & 2 \pm (A_2\otimes B_1 - A_2\otimes B_2) \geq 0, 
\label{XYdef}
\end{eqnarray}
with dichotomic observables $A_{1,2}\in\mathcal{A}$, $B_{1,2}\in\mathcal{B}$. 
One easily gets
\begin{eqnarray}
XY & = & 2 E_\mathrm{Bell} + (A_1A_2\otimes B_2B_1 - A_1A_2\otimes B_1B_2), \nonumber \\
YX & = & 2 E_\mathrm{Bell} + (A_2A_1\otimes B_1B_2 - A_2A_1\otimes B_2B_1).
\label{XYYX}
\end{eqnarray}

If the algebra $\mathcal{A}$ of Alice \emph{or} the algebra $\mathcal{B}$ of Bob 
is \emph{commutative},
the sum of the terms in brackets in (\ref{XYYX}) cancel and    
\begin{equation}
\label{BellQW}
Q_{\mathrm{AVR}} = \{X,Y\} = 4E_\mathrm{Bell}.
\end{equation}
As explained in Sec.\ \ref{QW}, since $[X,Y]=0$ and $X,Y\geq 0$, their symmetrized product 
must also be nonnegative: $Q_{\mathrm{AVR}}/4 =E_\mathrm{Bell}\geq0$. This is the Bell-CHSH inequality.

On the other hand, if the subalgebras $\cA$ and $\cB$ of Alice and Bob are \emph{both} noncommutative, one gets 
\begin{equation}
\label{BellQWcomm}
Q_{\mathrm{AVR}} = \{X,Y\} = 4E_\mathrm{Bell} - [A_1,A_2]\otimes [B_1,B_2],
\label{QAVR-A}
\end{equation}
which coincides with the result obtained in Ref.\ \cite{APVR} modulo a factor $2$. The above expression, with the choice of operators as in (\ref{bells}), turns out to be \emph{positive semidefinite} for any Bell state.
Therefore $Q_{\mathrm{AVR}}$ is not witnessing entanglement. 
$Q_{\mathrm{AVR}}$ can be shown to be negative for suitable factorized states, so it tests the quantumness of the individual subsystems.

\subsection{The Bell-CHSH inequality is also an anticommutator QW}
\label{BQW}
Let
\begin{eqnarray}
X & = & 2 \pm (A_1\otimes B_1 - A_2\otimes B_2) \geq 0, \nonumber \\
Y & = & 2 \pm (A_1\otimes B_2 + A_2\otimes B_1) \geq 0, \label{XYKP}
\end{eqnarray}
which are symmetric under the exchange $A \leftrightarrow B$, in contrast to those in (\ref{XYdef}).
Then 
\begin{eqnarray}
XY & = & 2 E_\mathrm{Bell} + [A_1,A_2] \otimes  \mathbb{I} + \mathbb{I} \otimes [B_1,B_2], \nonumber \\
YX & = & 2 E_\mathrm{Bell} - [A_1,A_2] \otimes  \mathbb{I}  - \mathbb{I} \otimes [B_1,B_2], 
\label{XYYXKP}
\end{eqnarray}
so that
\begin{equation}
\label{EWQW0}
Q_{\mathrm{AVR}} = \{X,Y\} = 4E_\mathrm{Bell}.
\end{equation}
This shows that the EW $E_\mathrm{Bell}$ is also an anticommutator QW\@:
if the Bell-CHSH inequality is violated by an entangled state $\rho$, then $\rho(Q_{\mathrm{AVR}}) < 0$\@.

An interesting remark is the following one: assume you have two particles, on which  Alice and Bob measure dichotomic observables. They put together their results and find that a state $\rho$ exists such that $\rho(E_{\mathrm{Bell}})<0$. Then they can conclude that their local observables do not commute.\footnote{In this case \emph{both} algebras $\cA$ and $\cB$ are noncommutative. Indeed, it is  easy to prove that if one of the two algebras were classical then any state $\rho$ of the composed system would necessarily be separable. See e.g.\  Prop.~2.5 in \cite{Keyl}.} In this sense, one can say that the Bell inequality is testing quantumness, and not simply entanglement: by looking only at the correlations of the two subsystems, one can check whether the two local (sub)algebras are noncommutative.

\subsection{A more general case}
Let us consider a more general case, i.e.\ the swap operator $S$ defined in (\ref{eq:swapop}) for a pair of qudits.
For a pair of qubits, it reads in the basis $\{\ket{00},\ket{01},\ket{10},\ket{11}\}$ (here $\ket{jk}:= \ket{j}\otimes\ket{k}$)
\begin{equation}
\label{EWQW}
S=\left(\begin{array}{cccc}
1&&&\\
&0&1&\\
&1&0&\\
&&&1
\end{array}\right),
\end{equation}
and for a pair of qutrits (in the basis $\{\ket{00},\ket{01},\ket{10},\ket{02},\ket{20},\ket{11},\ket{12},\ket{21},\ket{22}\}$),
\begin{equation}
\label{EWQW1}
S=\left(\begin{array}{ccccccccc}
1&&&&&&&&\\
&0&1&&&&&&\\
&1&0&&&&&&\\
&&&0&1&&&&\\
&&&1&0&&&&\\
&&&&&1&&&\\
&&&&&&0&1&\\
&&&&&&1&0&\\
&&&&&&&&1
\end{array}\right) .
\end{equation}
In this way, for a generic $d\times d$ system, $S$ is block-diagonal and contains two-by-two blocks
\begin{equation}
\left(\begin{array}{cc}
0&1\\
1&0
\end{array}\right)
\end{equation}
and these yield negative eigenvalues equal to $-1$.
The explicit construction of EWs of the form (\ref{wqsym}) for qudits is therefore reduced to understanding whether there exists a pair of positive operators $X$ and $Y$ ($X,Y\ge0$) such that
\begin{equation}
\left(\begin{array}{cc}
0&1\\
1&0
\end{array}\right)
\stackrel{?}{=}\{X,Y\}
\end{equation}
sector by sector.
It is trivial to construct $X$ and $Y$ for the diagonal elements $1$ in Eq.\ (\ref{EWQW}) or  (\ref{EWQW1}).

\subsubsection{$Q_{\mathrm{AVR}}$ for a generic two-state system.}
Let us first determine the eigenvalues of a QW of  the type (\ref{wqsym}) for a two-state system.
Generic positive operators $X$ and $Y$ of a two-state system can be expressed as
\begin{equation}
X=\frac{1}{2}\alpha(1+\bm{u}\cdot\bm{\sigma}),\qquad
Y=\frac{1}{2}\beta(1+\bm{v}\cdot\bm{\sigma}),
\end{equation}
with vectors $\bm{u}$ and $\bm{v}$, whose lengths are limited by
\begin{equation}
0\le u,v\le1,
\end{equation}
and with positive constants $\alpha,\beta>0$ (we exclude $\alpha,\beta=0$, since we are interested in nontrivial operators).
The anticommutator QW is then given by
\begin{equation}
Q_{\mathrm{AVR}}=\{X,Y\}
=\frac{1}{2}\alpha\beta[1+\bm{u}\cdot\bm{v}+(\bm{u}+\bm{v})\cdot\bm{\sigma}]
\label{eq:qwxyyx}
\end{equation}
and admits two eigenvalues
\begin{equation}
\lambda_\pm
=\frac{1}{2}\alpha\beta(1+\bm{u}\cdot\bm{v}\pm|\bm{u}+\bm{v}|),
\label{eqn:Lambdapm}
\end{equation}
the corresponding eigenstates $\ket{\lambda_\pm}$ being the two eigenstates of the operator $(\bm{u}+\bm{v})\cdot\bm{\sigma}$.

In order for (\ref{eq:qwxyyx}) to be a quantumness witness, one must have $\lambda_-<0$.
This entails
\begin{equation}
1+\bm{u}\cdot\bm{v}<|\bm{u}+\bm{v}|,
\end{equation}
which yields the condition
\begin{equation}
\cos^2\theta<\frac{u^2+v^2-1}{u^2v^2},
\label{eqn:Cond_QW_Qubit}
\end{equation}
where $\theta$ is the angle between the two vectors $\bm{u}$ and $\bm{v}$ and $\bm{u}\cdot\bm{v}=uv\cos\theta$.
See Fig.\ \ref{fig:QW_Qubit_Range}(a) for the range of $(u^2+v^2-1)/u^2v^2$.
See also Fig.\ \ref{fig:QW_Qubit_Range}(b), where the smallest attainable ratio $\lambda_-/\lambda_+$ is displayed as a function of $(u,v)$.
This shows that the $u=v=1$ case covers the widest range of eigenvalues $(\lambda_+,\lambda_-)$. 
\begin{figure}[b]
\begin{center}
\begin{tabular}{l@{\qquad}l}
(a)&(b)\\[-3.5truemm]
\includegraphics[height=0.32\textwidth]{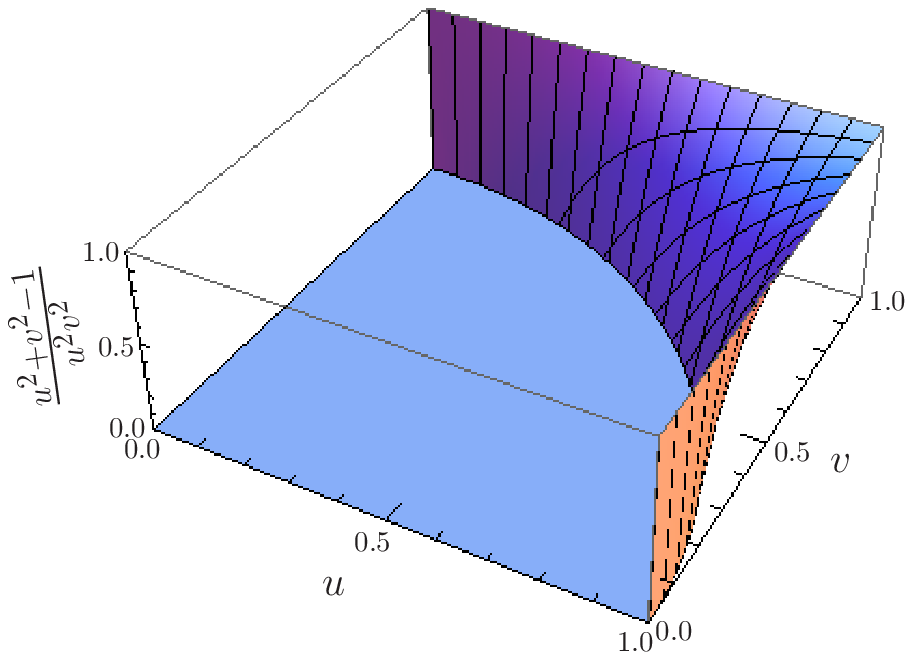}&
\includegraphics[height=0.32\textwidth]{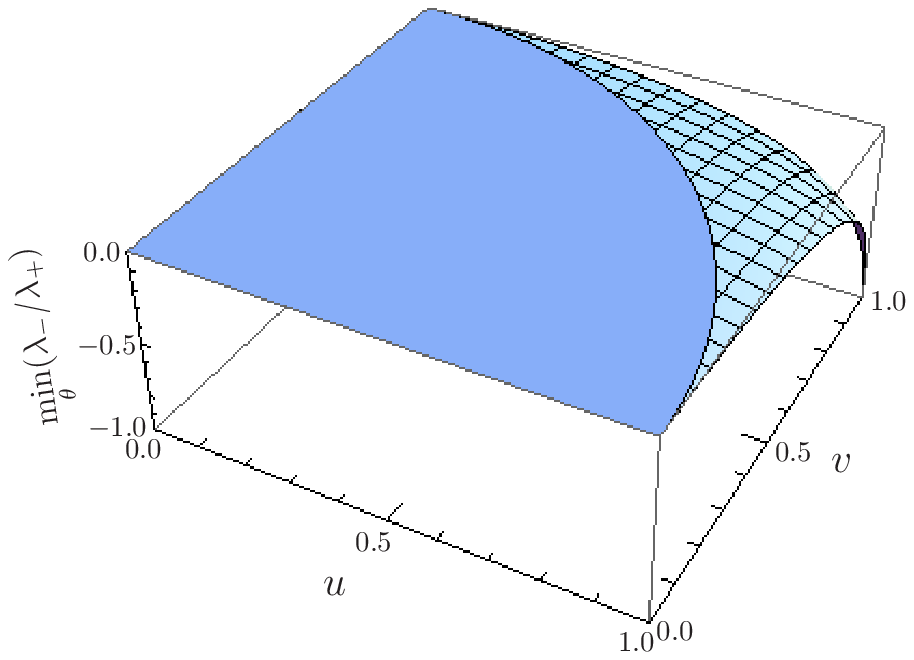}
\end{tabular}
\end{center}
\caption{(a) $(u^2+v^2-1)/u^2v^2$ as a function of $(u,v)$, yielding the upper bound on $\cos^2\theta$ for each $(u,v)$.
See (\ref{eqn:Cond_QW_Qubit}).
If $\theta$ violates this bound, $Q_\mathrm{AVR}=\{X,Y\}$ in (\ref{eq:qwxyyx}) is no longer a QW\@.
In the region where this upper bound is negative, $Q_\mathrm{AVR}$ can never be a QW\@. Therefore, only the positive range is shown.
(b) Smallest attainable ratio $\min_\theta(\lambda_-/\lambda_+)$ of the eigenvalues in (\ref{eqn:Lambdapm}) as a function of $(u,v)$. 
$\lambda_-/\lambda_+$ should be strictly negative for $Q_\mathrm{AVR}$ to be a QW\@.
Therefore, only the negative range is shown.
The ratio $\lambda_-/\lambda_+$ can come close to $-1$ only when $u=v=1$.}
\label{fig:QW_Qubit_Range}
\end{figure}

Let us thus fix $u=v=1$.
In this case, the two eigenvalues $\lambda_\pm$ of $Q_{\mathrm{AVR}}$ read
\begin{eqnarray}
\lambda_+=
2\alpha\beta\cos\frac{\theta}{2}\cos^2\frac{\theta}{4},
\quad
\lambda_-
=-2\alpha\beta\cos\frac{\theta}{2}\sin^2\frac{\theta}{4}
\qquad
(0<\theta<\pi).
\end{eqnarray}
Note that $\theta=0$ and $\pi$ are excluded, since $\lambda_-$ vanishes at these points.
The sum and ratio of the two eigenvalues read
\begin{equation}
\lambda_++\lambda_-=2\alpha\beta\cos^2\frac{\theta}{2},\qquad
\frac{\lambda_-}{\lambda_+}
=-\tan^2\frac{\theta}{4},
\label{eqn:SumRatioQ}
\end{equation}
respectively.
In particular, the ratio ranges between
\begin{equation}
-1<\frac{\lambda_-}{\lambda_+}<0.
\label{eqn:RatioQ}
\end{equation}

\subsubsection{$S$ is almost $Q_{\mathrm{AVR}}=\{X,Y\}$.}
Let us now consider a relevant two-by-two sector of $S$, and try to construct positive operators $X$ and $Y$ such that $S=\{X,Y\}$ \emph{in the sector}.
The eigenvalues of $S$ in each relevant two-by-two sector are $1$ and $-1$, whose ratio is $-1$.
Since the ratio of the eigenvalues $\lambda_\pm$ of the QW  (\ref{eq:qwxyyx}), for a two-state system, can only range between $-1<\lambda_-/\lambda_+<0$ as in (\ref{eqn:RatioQ}), there is no hope to construct $X$ and $Y$.
In this sense, the EW $S$ \emph{cannot} be written as an anticommutator QW\@.
However, if we add to $S$ a part proportional to the identity 
\begin{equation}
S\to \xi\mathbb{I}+S,
\end{equation}
the situation changes.
In such a case, the eigenvalues are shifted to
\begin{equation}
1+\xi\quad\textrm{and}\quad{-1+\xi}.
\end{equation}
Notice first that in order for this to remain an EW, $\xi$ should be bounded by $\xi<1$, 
otherwise we lose the negative eigenvalue and $\xi\mathbb{I}+S$ is no longer an EW\@.
In addition, from~(\ref{eqn:RatioQ}), the ratio of the shifted eigenvalues should be bounded by
\begin{equation}
-1<-\frac{1-\xi}{1+\xi}<0,
\end{equation}
in order for $\xi\mathbb{I}+S$ to be expressed as an anticommutator, $\xi\mathbb{I}+S=\{X_\xi,Y_\xi\}$.
This requires $\xi > 0$.
Therefore, $\xi$ should be bounded by
\begin{equation}
0<\xi<1,
\end{equation}
in order for $\xi\mathbb{I}+S$ to be an EW and at the same time an anticommutator QW\@.

Let us construct $X_\xi$ and $Y_\xi$ explicitly, with the lengths of the associated vectors $\bm{u}$ and $\bm{v}$ being $u=v=1$.
The angle $\theta$ between the two vectors $\bm{u}$ and $\bm{v}$ is fixed by the condition
\begin{equation}
\frac{\lambda_-}{\lambda_+}
=-\frac{1-\xi}{1+\xi}
=-\tan^2\frac{\theta}{4}.
\end{equation}
See (\ref{eqn:SumRatioQ}).
Hence, 
\begin{equation}
\xi=\cos\frac{\theta}{2},
\label{eqn:xi}
\end{equation}
and 
\begin{equation}
\left\{\begin{array}{l}
\medskip
\displaystyle
\bm{u}=(\sqrt{1-\xi^2}\cos\varphi,\sqrt{1-\xi^2}\sin\varphi,\xi),\\
\displaystyle
\bm{v}=(-\sqrt{1-\xi^2}\cos\varphi,-\sqrt{1-\xi^2}\sin\varphi,\xi),
\end{array}
\right.
\end{equation}
where $\varphi$ is an arbitrary parameter $0\le\varphi<2\pi$, and the $z$ direction is chosen in the direction of $\bm{u}+\bm{v}$, i.e.\ %
\begin{eqnarray}
\sigma_z
=\frac{\bm{u}+\bm{v}}{|\bm{u}+\bm{v}|}\cdot\bm{\sigma}
=\ket{\lambda_+}\bra{\lambda_+}
-\ket{\lambda_-}\bra{\lambda_-},
\nonumber\\
\sigma_x=\ket{\lambda_+}\bra{\lambda_-}
+\ket{\lambda_-}\bra{\lambda_+},\qquad
\sigma_y=-i\ket{\lambda_+}\bra{\lambda_-}
+i\ket{\lambda_-}\bra{\lambda_+}.
\label{eqn:SigmasRep}
\end{eqnarray}
See Fig.\ \ref{fig:uv}.
\begin{figure}
\begin{center}
\includegraphics[width=0.5\textwidth]{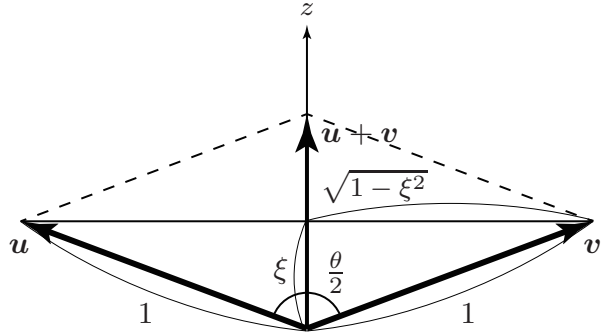}
\end{center}
\caption{Arrangement of $\bm{u}$ and $\bm{v}$ for $\xi\mathbb{I}+S=\{X_\xi,Y_\xi\}$.}
\label{fig:uv}
\end{figure}
On the other hand, by plugging (\ref{eqn:xi}) into the first relation in (\ref{eqn:SumRatioQ}), one has
\begin{equation}
\lambda_++\lambda_-
=2\xi
=2\alpha\beta\cos^2\frac{\theta}{2}
=2\alpha\beta\xi^2,
\end{equation}
which yields
\begin{equation}
\alpha\beta=\frac{1}{\xi}.
\end{equation}
Therefore, by~(\ref{eqn:SigmasRep}) we get
\begin{eqnarray}
X_\xi,Y_\xi
=\frac{1}{2\sqrt{\xi}}\,\biggl[&&
(1+\xi)
\ket{\lambda_+}\bra{\lambda_+}
+(1-\xi)\ket{\lambda_-}\bra{\lambda_-}
\\
&&{}
\pm\sqrt{1-\xi^2}
(
e^{-i\varphi}\ket{\lambda_+}\bra{\lambda_-}
+e^{i\varphi}\ket{\lambda_-}\bra{\lambda_+}
)
\biggr].
\end{eqnarray}
In particular, for a pair of qubits, we have $\ket{\lambda_\pm}=(\ket{01}\pm\ket{10})/\sqrt{2}$, 
so that
\begin{equation}
\xi\mathbb{I}+S
=\left(\begin{array}{cccc}
\xi+1&0&0&0\\
0&\xi&1&0\\
0&1&\xi&0\\
0&0&0&\xi+1
\end{array}\right),
\end{equation}
and
\begin{equation}
\fl
X_{\xi},Y_{\xi}
=\left(\begin{array}{cccc}
\medskip
\displaystyle\sqrt{\frac{1+\xi}{2}}&0&0&0\\
\medskip
0&\displaystyle\frac{1\pm\sqrt{1-\xi^2}\cos\varphi}{2\sqrt{\xi}}&
\displaystyle\frac{\xi\pm i\sqrt{1-\xi^2}\sin\varphi}{2\sqrt{\xi}}&0\\
\medskip
0&\displaystyle\frac{\xi\mp i\sqrt{1-\xi^2}\sin\varphi}{2\sqrt{\xi}}&
\displaystyle\frac{1\mp\sqrt{1-\xi^2}\cos\varphi}{2\sqrt{\xi}}&0\\
0&0&0&\displaystyle\sqrt{\frac{1+\xi}{2}}
\end{array}\right).
\end{equation}
Therefore, we get
\begin{equation}
S = \{X_{\xi}, Y_{\xi}\} - \xi \mathbb{I}, \qquad 0<\xi<1. 
\end{equation}
Since the positive shift $\xi$ can be  arbitrarily small, the EW $S$ is almost (but not quite) an anticommutator QW\@.
$X_\xi$ and $Y_\xi$ can be constructed in the same manner for higher-dimensional systems, sector by sector.

\section{Conclusions and perspectives}
\label{concl}
We have discussed the notions of quantumness and entanglement, showing that every entanglement witness is also a quantumness witness. Although entanglement is clearly a genuine quantum feature, our analysis makes use of strict mathematical definitions of witnesses. This enables one to put (physical) intuition on firm mathematical grounds. In turn, theorems and their derivations disclose alternative viewpoints:
we observed in Sec.\ \ref{BQW} that the Bell inequality, written as a QW, tests the ``global" quantumness of the composed system. This enables one to look at the Bell inequality from a novel perspective.

An interesting aspect that could be investigated in the future is whether the combined notions of quantumness and entanglement witnesses could shed light on the elusive notion of bound entanglement \cite{horod-a,horod-b}, for which the PT criterion does not apply.

We conclude by noting that the links between nonclassicality and entanglement have been also investigated in quantum optics \cite{vogel}.  At the root of this approach there is the idea that the partial transpose of some positive operators can detect nonclassicality in the light fields, witnessing it through suitable Cauchy-Schwarz inequalities \cite{chendeng}. Finally, the notion of ``partial" quantumness/classicality, suitably defined in order to provide a finer graining between the quantum and the classical worlds, can be shown to incorporate an ordering relation among different nonclassical correlations and entanglement measures \cite{PA}. In this analysis, that makes extensive use of inequalities, focus is on the quantum-to-classical transition.
Although the language used in the above-mentioned investigations is slightly different from that used in this article, the links between these methods are worth exploring in the future.

\ack 
We thank R.\ Fazio and A.\ J.\ Leggett for interesting remarks and useful suggestions.
This work is partially supported by the Joint Italian-Japanese Laboratory on ``Quantum Technologies: Information, Communication and Computation" of the Italian Ministry of Foreign Affairs.
P.F.\ acknowledges support through the project IDEA of the University of Bari.
K.Y.\ is supported by the Program to Disseminate Tenure Tracking System and the Grant-in-Aid for Young Scientists (B) (No.\ 21740294) both from the Ministry of Education, Culture, Sports, Science and Technology, Japan.
P.F., S.P.\ and K.Y.\ would like to thank the Centre for Quantum Technologies of the National University of Singapore for the kind hospitality.


\section*{References}
\begin{thebibliography}{10}

\bibitem{bell} Bell~J~S 1964 \textit{Physics} \textbf{1} 195

\bibitem{LG}
Leggett~A~J and Garg~A 1985 \textit{Phys. Rev. Lett.} \textbf{54} 857

\bibitem{Leggett-JPC2002}
Leggett~A~J 2002 \textit{J. Phys.: Condens. Matter} \textbf{14} R415

\bibitem{AVR} Alicki~R and Van~Ryn~N 2008 \textit{J. Phys. A: Math. Theor.} \textbf{41} 062001

\bibitem{APVR} Alicki~R, Piani~M and Van~Ryn~N 2008 \textit{J. Phys. A: Math. Theor.} \textbf{41} 495303

\bibitem{brida08} Brida~G, Degiovanni~I~P, Genovese~M, Schettini~V, Polyakov~S~V and Migdall~A 2008 \textit{Optics Express} \textbf{16} 11750

\bibitem{bridagisin} Brida~G, Degiovanni~I~P, Genovese~M, Piacentini~F, Schettini~V, Gisin~N, Polyakov~S~V and Migdall~A 2009  \textit{Phys. Rev.} A \textbf{79} 044102

\bibitem{Araki}
Araki~H 2000 \textit{Mathematical Theory of Quantum Fields} (Oxford: Oxford University Press)

\bibitem{BratteliRobinson}
Bratteli~O and Robinson~D~W 2002 \textit{Operator Algebras and Quantum Statistical Mechanics} vols 1-2 2nd edn (Berlin: Springer)


\bibitem{entanglementrev}
Amico~L, Fazio~R, Osterloh~A and Vedral~V 2008 \textit{Rev. Mod. Phys.} \textbf{80} 517

\bibitem{entrevhoro}
Horodecki~R, Horodecki~P, Horodecki~M and Horodecki~K 2009 \textit{Rev. Mod. Phys.} \textbf{81} 865

\bibitem{Keyl}
Keyl~M 2002   \textit{Phys. Rep.} \textbf{369} 431

\bibitem{horod-a}
Horodecki~P 1997 \textit{Phys. Lett.} A \textbf{232} 333
\bibitem{horod-b}
Horodecki~M, Horodecki~P and Horodecki~R 1998 \textit{Phys. Rev. Lett.} \textbf{80} 5239

\bibitem{vogel}
Kiesel~T, Vogel~W, Hage~B and Schnabel~R 2011
\textit{Phys. Rev. Lett.} \textbf{107} 113604

\bibitem{chendeng}
Chen J-L and Deng D-L 2009 \textit{Phys. Rev.} A \textbf{79} 012115

\bibitem{PA} Piani~M and Adesso G 2011 
arXiv:1110.2530 [quant-ph]

\end{thebibliography}
\end{document}